\documentclass[preprint]{imsart}

\usepackage{amsthm,amsmath,bbm}
\usepackage{hyperref}

\usepackage{xcolor,listings,multicol} 
\usepackage{amssymb} 
\usepackage[]{enumitem}

\usepackage{geometry} 
\startlocaldefs

\theoremstyle{plain}
\newtheorem{theorem}{Theorem}[section]
\newtheorem{lemma}[theorem]{Lemma}
\newtheorem{corollary}[theorem]{Corollary}

\theoremstyle{definition}
\newtheorem{definition}{Definition}
\newtheorem{example}{Example}
\newtheorem{remark}{Remark}

\endlocaldefs

\begin{document}

\begin{frontmatter}

\title{Generation of discrete random variables in scalable frameworks}
\runtitle{Discrete random variables in scalable frameworks}

\author{\fnms{Giacomo} \snm{Aletti}\corref{}\ead[label=e1]{giacomo.aletti@unimi.it}\thanksref{t1}}
\address{ADAMSS Center, Universit\`a degli Studi di Milano, \\
20131 Milano, Italy
\\ \printead{e1}}
\thankstext{t1}{Member of ``Gruppo Nazionale per il Calcolo Scientifico (GNCS)'' 
of the Italian Institute ``Istituto Nazionale di Alta Matematica (INdAM)''.
This work was partially developed during a visiting
research period at the Volgenau School of Engineering,
George Mason University, Fairfax (VA), United States. It was revised
during a visiting research period at the School of Mathematical Science,
Fudan University, Shanghai, China. The author thanks for the hospitality.} 

\runauthor{Giacomo Aletti}

\begin{abstract}
In this paper, we face the problem of simulating discrete random variables with general and varying distributions in a scalable framework, where fully parallelizable operations should be preferred.
The new paradigm is inspired by the context of discrete choice models. Compared to classical algorithms, we add parallelized randomness,  and we leave the final simulation of the random variable to a single associative operation.
We characterize the set of algorithms that work in this way, and those algorithms that may have an additive or multiplicative local noise. As a consequence, we could define a natural way to solve some popular simulation problems.
\end{abstract}



\end{frontmatter}

\section{Introduction}

The aim of this paper is to define and to characterize a new method 
for the generation of discrete random variables in a scalable framework.
This is done by merging two apparently different fields, 
namely the discrete random variables generation
and the discrete choice framework. 

The generation of discrete random variables 
may be made in different ways, see 
\cite{Devroye86,Fast2004,Rub16,DRNG_2013} 
and the references therein. 
The most popular  idea is to invert the cumulative function $F_k$ defined on the sets of the indexes $k=1,\ldots,n$
of the support. When the cumulative function is not parametrized,
we recall that a bisection search
takes $O(\log_2(n))$ comparison to invert $F$ (see \cite[Section III.2.4]{Devroye86}), once we have computed and stored
the table $\{(k,F_k),k=1,\ldots,n\}$ properly. A traditional linear search may be done with $O(n)$ comparison.

There are more sophisticated ways to invert a discrete distribution.
They typically require other precomputations and bookkeepings.
We recall here three fast popular methods. Even if these three methods might be dated,
the most recent books do not provide other paradigms that generate directly discrete random variables 
with general distributions (cfr., for example, \cite{Rub16}). 

The \emph{table look-up method} is the fastest way to simulate a large number of i.i.d.\ discrete random variables
(see  \cite[Section III.3]{Devroye86} and \cite{Fast2004}),
when all the probabilities $\{p_j = \frac{n_j}{m} ,j=1,\ldots,n\}$ are rational numbers 
with the same denominator $m$.
Obviously, $n_1+n_2+\cdots+n_N = m$, and hence we may set up a table $\{(k,Q_k),k=1,\ldots, m\}$ such that
$n_j$ distinct values of $Q$'s are set to $j$, for each $j=1,\ldots,n$. To simulate the required random variable,
take an integer uniform discrete random variable $Y$ on $\{1,\ldots,m\}$ and then compute $Q_Y$ in $O(1)$ time.
The drawback of this method compared to the previous ones is the amount of required space, that is necessary to store
the exact pseudo-inverse function $Q_j = F^{-1}(\frac{j}{m})$ on the equispaced nodes $\{\frac{1}{m} ,\ldots, 1\}$.

The \emph{method of guide tables} was introduced by \cite{chen74}, and stores 
a second ``guiding table''  $\{(k,G_k),k=1,\ldots,m\}$ that helps the generation, by reducing the
expected number of comparison to less than $ 1 + \tfrac{n}{m}$, see \cite[Section III.3.4]{Devroye86}.

The last method is called \emph{alias table} (see \cite[Section III.4]{Devroye86}). It was firstly
introduced in \cite{walker74,walker77}, and then it was improved together with a simple probabilistic proof in \cite{KroPet79}.
This method does not need the computation of the cumulative function even if it requires the probabilities of the events to be normalized. 
Besides this, it uses a special table build in $O(n)$ time (called alias table).
During the simulation process, it uses only $2$ comparisons.

In all these methods, the simulation of a random variable is made by constructing 
a table based on the \emph{probabilities} of the possible events. In many applied situations,
the probabilities are computed up to a multiplicative constant, or in logarithm scale up to the translational constant.
Accordingly, a preprocess must be done to reconstruct the normalized probabilities before using any of the above methods.
Moreover, if one has to simulate random variables with different distributions, it must be allocated a different table for each of them.
Finally, if the probability of one of the events changes (or if the support itself changes), one must restart all the process.

The novelty of this paper is the introduction of a new method for the generation of several discrete
independent random variables with possible different distributions, and 
whose distributions does not belong to a parametric family. A key point is the fact that this method 
does not precompute a table based on the probability of the events. Instead, it is
%
based on the following two assumptions:
\begin{itemize}
\item first, for each random variable, and for each point of the support, 
we can perform an action (called utility) that involves the sole local accessible information. 
This operation is hence fully parallelizable, and we do not need to take care of scaling constants 
when the probabilities are given in logarithm space and/or up to a constant;
\item secondly, for each random variable, a single associative operation on the utilities on the points 
of its support must  finally simulate  the discrete random variable.
\end{itemize}
This new method may be also be updated in a fast natural way when the probabilities 
are changing with time
(one local update and an associative operation),
and hence it may also be used in real time problems. 
The main results of this paper is the characterization of all the possible
ways of simulating a discrete random variables with such assumptions.

Obviously, there is a counterpart. On the one hand, in fact, we do not provide a precomputed table 
and we perform only one associative operation
on some locally calculated quantities. On the other hand, each of these quantities
depends on an independent source of uncertainty. Summing up,
we increment the total amount of randomness (by adding a local source), 
and hence we could reduce the non-local operations to an associative one. 

The idea behind this new method may found a counterpart in 
the framework of discrete choice models,
where the point of view is to understand the behavioral process that leads 
to an agent's choice among a set of possible actions,
see \cite{Kenneth09} for a recent book on this subject. 
The researcher knows the set of the possible actions, and by
observing some factors, he may infer something about the agent's preferences. 
At the same time, he cannot observe other random factors, linked to each possible action, that cause the final decision. 
If the researcher could have observed these hidden factors, he could have predicted the action 
chosen by the agent by selecting 
the one with maximum utility function. 

The process of choice selection has the two characteristics we gave above for random
generation: 
\begin{itemize}
\item for each agent (random variable), for each action (point of the support), 
the utility function -a given deterministic function of the observed and the hidden factors- is calculated and depends only on local variables;
\item for each agent, the final choice is made by selecting the action with maximum utility value (associative operation).
\end{itemize}
In other words, we are changing 
the usual point of view belonging to discrete choice framework to produce 
and characterize new scalable
simulators for discrete random variables, 
based on primary functions given, e.g., 
in a general \texttt{SQL} database.
As a by product, 
we will be able to characterize 
all the choice models with independent and
identically distributed hidden factors and such that the probability
of choosing an action is (proportional to) a given function of the observable factors.

The content of the paper is structured as follows. In the Example~\ref{exa:dB} of 
Section~\ref{sec:examples}, we introduce a very general problem of randomization in classification procedures
in a \texttt{SQL} environment, and we solve it with our new paradigm.
The subsequent Example~\ref{exa:dcm} shows the mathematical position of the same problem in discrete choice's 
framework. The reason why the two examples shares the same problem is discussed at the end of the examples.
The section ends with a discussion on the fast updating process that is required when the distributions vary with time.

In Section~\ref{sec:th_res} we give the main results of the paper,
based on the notion of $\max$-compatible family of distributions, which is the mathematical structure at the base of
our new method. 
This family is fully characterized in terms of the cumulative functions in Theorem~\ref{thm:main_repr}, 
that can be seen as the main mathematical result of this paper.
The section continues with the description of the new algorithm of 
random variable generation in terms of max-compatible families, 
and it ends with
the characterization of some natural models 
that may be found in 
usual applied situations.

Section~\ref{sec:concl} concludes the paper with some future research regarding this topic.

All the proofs of the results are referred to the appendix, 
and they are preceded by some general theoretical results on 
real continuous distributions.

\section{Motivating examples}\label{sec:examples}

In this section, we show two examples. In the first one, we show how the new method may be used to randomize
a Bayesian classifier in a \texttt{SQL} environment. In this example, the data are stored 
in a table called \texttt{Mytable}, with columns \texttt{ID, QUAL, Strength}.
The column \texttt{ID} identifies $N$ different users (or documents, or images, \ldots), 
\texttt{QUAL} refers to a quality of the user,
\texttt{Strength} is a real number that exhibits how much the quality \texttt{QUAL} is expressed
by the user \texttt{ID}. Note that different users may espress different qualities with different strength;
we only assume that in \texttt{Mytable} there are not two rows with the same couple \texttt{ID, QUAL}.

In the second example, we give a perspective of our new paradigm in terms of the \emph{discrete choice model}'s
framework. In this framework, the following objects are defined:
\begin{itemize}
\item the \emph{choice set} $\{x_l, l = 1, \ldots\}$, which is the set of options that are available to the $N$ decision makers;
\item the \emph{consumer utility law}, which is a function that assigns to each decision maker and each option
the utility that each decision will bring to the player. 
Here, we assume that the utility laws of different players are independent among each other.
We label the decision maker by $t\in\{1,\ldots,N\}$, and we denote by $x_{j}$ its $j$-th alternative  
among a set of finite number of alternatives $\{x_{1}, \ldots , x_{n}\}$ in the
choice set. 
The utility $X_{t\,j}= \hat{g}(s(t,x_j) , \epsilon_{t\,j})$ is based on two parts:
\begin{enumerate}
\item the first one, labeled $s(t, x_j)$, that is known;
\item the second part $\epsilon_{t\,j}$, random. 
Here $\{\{\epsilon_{t\,1}, \ldots , \epsilon_{t\,n}\}, t=1,\ldots,N\}$
are independent families of independent random variables, all with a common law $\epsilon$ that does not depend on $j$ and $t$;
\end{enumerate}
\item the \emph{choice} $\hat{x}_n$ of the $n$-th decision maker, derived from utility-maximizing procedure:
 \[
 \hat{x}_t= \mathop{\arg\max}_{\{x_{1}, \ldots , x_{n}\}} X_{t\,j} = 
 \mathop{\arg\max}_{\{x_{1}, \ldots , x_{n}\}} \hat{g}(s(t,x_j) , \epsilon_{n\,j});
 \]
\item the \emph{choice probabilities} $\{p_{t\,j}, j = 1,\ldots, n\}$, derived from utility-maximizing behavior: 
\[
p_{t\,j} = P( \hat{x}_t = x_{j}) .
\] 
\end{itemize}

\begin{example}[Randomized classification]\label{exa:dB}
In a Multinomial naive Bayes classifier problem, each user is assumed to generate 
a sample $\mathbf{y}=(y_1, \ldots, y_M)$ that depends on the $\texttt{QUAL}$
that it is expressing. Each $y_{i}$ counts the number of times event $i$ was observed, and the joint
probability 
is
\[
P( \mathbf{y} | {\texttt{QUAL}=x} ) =  \frac{(\sum_i y_i)!}{\prod_i y_i !} \prod_i {P(i|{\texttt{QUAL}=x}) }^{y_i},
\]
where ${P(i|{\texttt{QUAL}=x}) }$ is the probability that event $i$ occurs under $\texttt{QUAL}=x$.
This is the event model typically used for document classification, see, e.g.\ \cite{Hand01,Rennie03}.

The multinomial naive Bayes classifier computes \texttt{Strength} as the log-likelihood function up to a constant:
\begin{equation}\label{eq:defStrength}
\log P( {\texttt{QUAL}=x} | \mathbf{y} ) = \text{const} + \underbrace{\log P( {\texttt{QUAL}=x} ) +
\sum_i y_i \log P(i|{\texttt{QUAL}=x})}_{\text{\texttt{Strength} of \texttt{QUAL}$=x$}} ,
\end{equation}
and then it selects, for each user, the quality \texttt{QUAL} with the higher \texttt{Strength}.
Of course, one could select a \texttt{QUAL} randomly from each user with the same probability, 
just adding a random column 
\texttt{RND} generated by \texttt{RAND()} to \texttt{Mytable}, and use it for the selection instead of
\texttt{Strength}. 
Among a lot of equivalent expressions, once we have set a proper index on \texttt{Mytable}, 
the two procedures might be set directly in a fast \texttt{SQL} query, see Listing~\ref{lst:SQLclass} and
Listing~\ref{lst:SQLrand}.


\noindent\begin{minipage}{.45\textwidth}
\begin{lstlisting}[
           language=SQL,
basicstyle=\footnotesize\ttfamily,
numbers=left,
numbersep=10pt,
xleftmargin=20pt,
frame=tb,
framexleftmargin=20pt,
caption={\texttt{SQL} code for multiclass classifier} \label{lst:SQLclass}
        ]
SELECT a.ID, a.QUAL
FROM MyTable a
LEFT OUTER JOIN MyTable b
    ON a.ID = b.ID
     AND a.Strength < b.Strength
WHERE b.id IS NULL;
\end{lstlisting}
\end{minipage}\hfill
\begin{minipage}{.45\textwidth}
\begin{lstlisting}[
           language=SQL,
basicstyle=\footnotesize\ttfamily,
numbers=left,
numbersep=10pt,
xleftmargin=20pt,
frame=tb,
framexleftmargin=20pt,
caption={\texttt{SQL} code for uniformly random selection} \label{lst:SQLrand}
        ]
SELECT a.ID, a.QUAL
FROM MyTable a
LEFT OUTER JOIN MyTable b
    ON a.ID = b.ID
     AND a.RND < b.RND
WHERE b.id IS NULL;
\end{lstlisting}
\end{minipage}

This last random selection shows, in fact, a paradigm to generate a random variable different from the usual ones.
In fact, for each point of the support, the algorithm generates the uniform random variable \texttt{RND}
(local fully parallelizable action),
independently of the rest, and then it selects the point with the higher \texttt{RND} (single associative operation).

The question of this paper is how to select a \texttt{QUAL} randomly from each user with a probability 
proportional to $\exp(\texttt{Strength})$, or, more generally, $f(\texttt{Strength})$,
where $f:\mathbb{R}\to\mathbb{R}_+$ is a given non-negative function.
It is obvious that, when $f(x) = 1$ (or any other constant $c$) for any $x$, 
the new selection should return a procedure equivalent to Listing~\ref{lst:SQLrand}. 

\begin{table}
\caption{Generation of four discrete random variables, with different distributions,
within the settings of Example~\ref{exa:dB}.
Left table: example of \texttt{Mytable}, where
the quantity \texttt{RND2} is computed in fully parallelized way: it 
depends on the corresponding \texttt{Strength} and the result of the \texttt{SQL} function \texttt{RAND()}. 
Right: \texttt{SQL} code.
Bottom: the generation of the random variables is made by selecting the \texttt{QUAL} with the higher
\texttt{RND2}, for each \texttt{ID}.}\label{tab:codeSQL}

\smallskip

\begin{tiny}
\noindent\begin{minipage}{.45\textwidth}
\centering
\begin{tabular}{|c|c|c|c|}
\hline
\texttt{ID}&
\texttt{QUAL}&
\texttt{Strength}&
\texttt{RND2}\\
\hline
\#1 & YELLOW & -1 & 0,664834081 \\
\#4 & PURPLE & -4 & -4,426142579 \\
\#1 & WHITE & 2 & 2,926653411 \\
\#1 & RED & 2 & 5,612483956 \\
\#3 & CYAN & -1 & -2,501775035 \\
\#4 & WHITE & -3 & -3,131509289 \\
\#3 & WHITE & 0 & 0,524126732 \\
\#2 & RED & 1 & 1,30338907 \\
\#4 & YELLOW & 1 & 3,083588566 \\
\#1 & ORANGE & 5 & 5,603956288 \\
\#4 & CYAN & 0 & 1,66402363 \\
\#2 & WHITE & 4 & 4,143186699 \\
\#2 & CYAN & 5 & 3,77108384 \\
\#2 & ORANGE & 0 & 1,682024182 \\
\hline
\end{tabular}
\end{minipage}
\end{tiny}
\hfill + \hfill
\begin{minipage}{.45\textwidth}
\begin{lstlisting}[
           language=SQL,
basicstyle=\footnotesize\ttfamily,
numbers=left,
numbersep=10pt,
xleftmargin=20pt,
frame=tb,
framexleftmargin=20pt,
caption={\texttt{SQL} code for discrete random variable generation based on \texttt{Mytable}}
\label{lst:SQL_multRand}
       ]
SELECT a.ID, a.QUAL
FROM MyTable a
LEFT OUTER JOIN MyTable b
    ON a.ID = b.ID
     AND a.RND2 < b.RND2
WHERE b.id IS NULL;
\end{lstlisting}
\end{minipage}
\\[.2cm]
\centering
$\Downarrow$  \\[.2cm]
\centering
\begin{footnotesize}
\begin{tabular}{|c|c|}
\hline
\texttt{ID}&
\texttt{QUAL}\\
\hline
\#1 & RED \\
\#2 & WHITE \\
\#3 & WHITE \\
\#4 & YELLOW \\
\hline
\end{tabular}
\end{footnotesize}
\end{table}

To solve this problem, 
the idea is to merge the information in the two codes above. First, we define 
a new column \texttt{RND2}  (see Table~\ref{tab:codeSQL}) given by
${g}( f(\texttt{Strength}),\texttt{RND})$, where $g:\mathbb{R}_+\times (0,1) \to \mathbb {R}$
is a suitable function. This procedure is fully parallelizable and scalable.
Then, the solution of the problem will be performed with the code in Listing~\ref{lst:SQL_multRand}
which selects, for each user, the quality \texttt{QUAL} with the higher \texttt{RND2}.
Of course, the column \texttt{RND2} (defined by $g$) 
must produce the desired result. In other words, if a user expresses the three qualities $\texttt{QUAL}_1$,
$\texttt{QUAL}_2$ and $\texttt{QUAL}_3$ with corresponding streghts $\texttt{Strength}_1$,
$\texttt{Strength}_2$ and $\texttt{Strength}_3$ we must be sure, for example, that the first
quality is selected proportionally to $e^{\texttt{Strength}_1}$, that is
\begin{multline}\label{exa1:prob}
P( {g}( e^{\texttt{Strength}_1},U_1) > \max({g}( e^{\texttt{Strength}_2},U_2),
{g}( e^{\texttt{Strength}_3},U_3) ) 
= \frac{e^{\texttt{Strength}_1}}{
e^{\texttt{Strength}_1}+e^{\texttt{Strength}_2}+e^{\texttt{Strength}_3}},
\end{multline}
where $U_1,U_2,U_3$ are three independent uniform random variables.
To achieve this task, we will characterize in Theorem~\ref{thm:main_repr} all the functions
$Q_{f(\texttt{Strength})} (U) = {g}( f(\texttt{Strength}), U) $ for which 
the solution of the problem may be coded as in Listing~\ref{lst:SQL_multRand}.

%
\end{example}

\begin{remark}
When $f(x) = c$ for any $x$, the column $\texttt{RND2} = {g}( f(\texttt{Strength}),\texttt{RND}) $ is independent of 
$\texttt{Strength}$. Then the result of 
Listing~\ref{lst:SQL_multRand} is equivalent to that of Listing~\ref{lst:SQLrand},
if the probability of having
the same value for two different $\texttt{QUAL}$s is null. 
The fact that $\texttt{RND2}$ \emph{must} have
a continuous distribution (as the uniform \texttt{RAND()}) is proved in Lemma~\ref{lem:Falpha_cont}.
\end{remark}

\begin{example}[Probit model and choice probabilities in discrete choice framework]\label{exa:dcm}
With the notation of the  discrete choice model's framework given above,
we are interested here in characterizing all the common laws $\epsilon$ and 
the utility laws $X_{t\,j}= \hat{g}(s(t,x_j) , \epsilon_{t\,j})$ 
which give a preassigned choice probabilities $\{p_{t\,j}, j = 1,\ldots, n\}$. 
For example, when the law of $\epsilon$
is a Gumbel distribution and $\hat{g} ( s,u) = s+u$, the model is called \emph{probit}.
In this case, it is known (see \cite{Kenneth09}) that 
\begin{equation}\label{eq:probit}
p_{t\,j} = \frac{ \exp(s(t,x_j)) }{ \sum_{k=1}^n \exp(s(t,x_k)) }.
\end{equation}
\end{example}

Example~\ref{exa:dB} and Example~\ref{exa:dcm} are clearly linked: each user, identified by \texttt{ID} in the first example, 
is one of the $N$ decision makers in the second example. 
Each $x_j$ in Example~\ref{exa:dcm} represents a quality \texttt{QUAL} in Example~\ref{exa:dB}.
The known quantity $s(t, x_j)$ in Example~\ref{exa:dcm} is expressed by $\texttt{Strength}$ in  Example~\ref{exa:dB}.
The uniformly distributed random variable \texttt{RND} may be transformed into $\epsilon$ (and vice-versa,
as a consequence of Corollary~\ref{cor:Falpha_to_U}), so that
$X_{t\,j}$ in Example~\ref{exa:dcm} corresponds to \texttt{RND2} in Example~\ref{exa:dB}.

Notably, the equation \eqref{eq:probit} shows a possible solution for the randomized Bayesian classifier, where 
\texttt{Strength} is defined in \eqref{eq:defStrength}. In fact, it is sufficient to take
$\texttt{RND2} = \texttt{Strength} - \log(-\log(\texttt{RAND()}) )$, 
since $-\log(-\log(\texttt{RAND()}))$ is distributed as a Gumbel random variable (see also Remark~\ref{rem:probit-canonic}).

\subsection{Generation with distributions that vary with time}
Let us come back to the Example~\ref{exa:dB}, and suppose that the distributions 
vary during the time. 
We recall here that each row is identified by the couple $\texttt{ID} | \texttt{QUAL}$.
If we update the value of \texttt{Strength} and \texttt{RND2} 
in a row of \texttt{Mytable}, we are changing the probability of the corresponding event; 
if we add or remove
some rows that correspond to an \texttt{ID}, then we are changing the suppport of
its discrete random variable; if we add a row with a new \texttt{ID}, we are adding a new
random variable.

When we deal distributions that vary with time, it is convenient to store also the maximum value of \texttt{RND2} 
during the process of generation of the random variables. 
Accordingly, let us assume that the row $(\#1 | \textrm{RED} | 5,612483956)$
is present in the table at the bottom of Table~\ref{tab:codeSQL}. 

If a new query adds the row $(\texttt{ID}= \#1 |  \texttt{QUAL} = q |  \texttt{Strength}= s | \texttt{RND2}=r)$ 
in \texttt{Mytable}, 
or it updates the existing row 
to it,
then the generation of the random variable correspondent to $\texttt{ID}= \#1$
is changed according to the following table:

\medskip

\begin{tabular}{r|c|c}
& $q\neq \textrm{RED}$ & $q= \textrm{RED}$ \\
\hline
$r < 5,612483956$ & 
do nothing 
&
\multicolumn{1}{ c| }{select afresh the maximum for ${\texttt{ID}}= \#1$}
\\
\hline
$r>5,612483956$ & 
\multicolumn{2}{ c| }{update into the bottom table $(\#1 | q | r )$} 
\\ \cline{2-3}
\end{tabular}

\medskip

Note that the entire associative procedure is required only when
$q= \textrm{RED}$ and $r < 5,612483956$, and it is applied 
only to the subset with ${\texttt{ID}}= \#1$.
Of course, this task is also necessary if a query deletes from \texttt{Mytable}
the entire row $(\#1 | \textrm{RED} | 2 | 5,612483956)$. 
No updating process is required after the deletion of any
row $(\#1 |  \texttt{QUAL}|  \texttt{Strength}| \texttt{RND2})$, whenever $\texttt{QUAL}\neq \textrm{RED}$.

Finally, if a query adds a row with a new $\texttt{ID}$ to \texttt{Mytable}, the corresponding-generated
 $(\texttt{ID} | \texttt{QUAL} | \texttt{RND2})$ is immediately added to the simulation table.

\section{Theoretical and applied results}\label{sec:th_res}

In the sequel $F,F_X,F_{\alpha},\ldots$ will always denote cumulative 
distributions on $\mathbb R$,
while $X,X_{\alpha}, \ldots$ denote random variables on $\mathbb R$.
$X_1\sim X_2$ means that $X_1$ and $X_2$ share the same distribution, while 
$X\cong F$  means that the random variable $X$ has cumulative function $F$, also denoted
 by $F_X$. $U$ denotes always the random variable with uniform distribution on $(0,1)$. 
We will denote by $Q:(0,1) \to \mathbb{R}$ the quantile function associate to
a cumulative function $F$ in the following way:
\[
Q(u) = \inf\{ x \in\mathbb{R}\colon F(x)>u \} = \sup\{ y \in\mathbb{R}\colon F(y)\leq u \} .
\]

We now introduce the parametric family of probability distributions that are compatible with the associative operator ``max''. 
We require that the maximum value may be reached at each realization of any subsets of the family, proportionally to the parameters of the distributions that have been selected from the family and have generated the sample. 
\begin{definition}\label{defi:max-compatible}
Let $\mathcal C= \{F_\alpha,\alpha > 0\}$ be a parametric family of real probability distributions. 
The family $\mathcal C $ is called $\max$-compatible if, for any $n \geq 2$, whenever
$X_{\alpha_i} \cong F_{\alpha_i}, i=1,\ldots,n$ 
are independent random variables, we always have that
\begin{equation}\label{eq:def_max-compatible}
P\big( X_{\alpha_1} >\max(X_{\alpha_2},\ldots,X_{\alpha_n}) \big) = \frac{\alpha_1}{\sum_{i=1}^n\alpha_i} .
\end{equation}
The family $\mathcal C $ is called $\min$-compatible if 
$\max$ and $>$ are replaced by $\min$ and $<$  in \eqref{eq:def_max-compatible}.
\end{definition}
\begin{remark}
The key equation \eqref{eq:def_max-compatible} is the mathematical 
characterization of \eqref{exa1:prob} of Example~\ref{exa:dB}, with 
$\alpha_i = e^{\texttt{Strength}_i}$ and
$X_{\alpha_i}= {g}( e^{\texttt{Strength}_i},U_i)$.
\end{remark}
\begin{remark}\label{rem:cont}
When $X$ and $Y$ are independent, it is well known that $F_{\max(X,Y)}(t)=F_{X}(t)F_{Y}(t)$.
Since the product of continuous functions is continuous, 
the distribution function of $\max(X_{\alpha_2},\ldots,X_{\alpha_n}) $
is continuous whenever $X_{\alpha_2},\ldots,X_{\alpha_n}$ belong to a $\max$-compatible family,
by Lemma~\ref{lem:Falpha_cont}. 
It is hence possible to replace $>$ with $\geq$ in \eqref{eq:def_max-compatible}.
\end{remark}

We now state the following theorem, that characterizes all the $\max$-compatible families.
In particular, \ref{rap:c} ensures the associative property of the family 
and \ref{rap:d} characterizes the dependence of the cumulative functions with respect to the
parameter $\alpha$.

\begin{theorem}[Representation of $\max$-compatible families]\label{thm:main_repr}
Let $\mathcal C= \{F_\alpha,\alpha > 0\}$ be a parametric family of real continuous probability distributions. 
The following statements are equivalent:
\begin{enumerate}[label={(\alph*)}]
\item\label{rap:a} the family $\mathcal C$ is $\max$-compatible;
\item\label{rap:a1} for 
any monotone increasing function $h:\mathbb R\to \mathbb R$, the family 
\[
\mathcal C' = \{F'_\alpha(t) = F_{\alpha}(h(t)) , \alpha > 0\} , \qquad (\text{where } F_{\alpha} \in \mathcal C)
\] 
is $\max$-compatible;
\item\label{rap:b} the family 
\[
\mathcal C' = \{F'_\alpha(t) = 1-F_{\alpha}(-t) , \alpha > 0\} , \qquad (\text{where } F_{\alpha} \in \mathcal C)
\] 
is $\min$-compatible;
\item\label{rap:c} whenever
$X_{\alpha_1} \cong F_{\alpha_1}$ and 
$X_{\alpha_2} \cong F_{\alpha_2}$ 
are independent random variables, we always have that
\begin{equation}\label{eq:def_compat}
X_{\alpha_1+\alpha_2} \sim \max(X_{\alpha_1} , X_{\alpha_2} ),
\qquad \text{where $X_{\alpha_1+\alpha_2} \cong F_{\alpha_1+\alpha_2}$};
\end{equation}
\item\label{rap:d} for any $\alpha >0$, $F_\alpha (t) = (F(t))^\alpha$, 
where $F$ is any continuous cumulative distribution function, whence $F= F_1$;
\item\label{rap:d2} there exists a strictly increasing quantile function $Q:(0,1)\to\mathbb{R}$
such that, for any $\alpha >0$, $X_\alpha \sim Q(\sqrt[\alpha]{U})= Q_{\alpha}(U)$, where
$U$ is a $(0,1)$-uniformly distributed random variable (and hence $Q= Q_1$).
\end{enumerate}
\end{theorem}

\subsection{Generation of discrete random variables}
The conditions \ref{rap:d} and \ref{rap:d2} in Theorem~\ref{thm:main_repr} characterizes 
the cumulative functions $F_\alpha$ and the quantile functions $Q_\alpha$ 
of any $\max$-compatible family, in terms of 
the cumulative $F_1$ and quantile $Q_1$ functions, that can be freely chosen.
In particular, the family $\mathcal C= \{t^\alpha\mathbbm{1}_{(0,1)}(t),\alpha > 0\}$ may be seen as
`the canonical one', since it is build starting from the uniform distribution. In this case, if 
$X_\alpha \cong F_\alpha(t) = t^\alpha\mathbbm{1}_{(0,1)}(t)$, then $X$ may be generated by setting
$X_\alpha = Q_\alpha(U) = \sqrt[\alpha]{U}$, with $U$ uniform. 
\begin{remark}\label{rem:probit-canonic}
If we take $f(\texttt{Strength}) = \alpha$ and $g(\alpha,u) = Q_1(\sqrt[\alpha]{u})$, then
$\widehat{\texttt{RND2}} = (\texttt{RAND()})^{\frac{1}{\alpha}} $ is the canonical solution 
with the code in Listing~\ref{lst:SQL_multRand}.
The solution 
$$
{\texttt{RND2}} = \texttt{Strength} - \log(-\log(\texttt{RAND()}) ) = - \log(-\log( \widehat{\texttt{RND2}})),
$$ 
given at the end of Section~\ref{sec:examples} for $\alpha = e^{\texttt{Strength}}$, 
is based on a monotone transformation that does not change 
the selection of the maximum point.
\end{remark}



From a computational point of view, it must be underlined that both the operations
$\alpha = f(\texttt{Strength})$ and $\sqrt[\alpha]{u}$ may lead to unexpected precision errors.
The freedom in choosing the quantile function $Q_1$ helps us to face this problem.
One the one hand, it may transform the problem on a different scale, and thus avoiding the 
transformation of \texttt{Strength}. On the other hand, it will imply the transformation
of the uniform random variable $U$. This last operation may be done sometimes in 
a fast and ad hoc way (see, e.g., \cite{Marsaglia00}).
For the first purpose, we now underline some ``special families'' of distributions. 
The first one is useful when one records \texttt{Strength} as a linear
transformation of $\log \alpha$, as in Example~\ref{exa:dB}, and 
exponentiating it may cause errors.
The other two families deal with records of the order of $\alpha^c$ and of $\alpha^{-c}$. 
The functional forms of $Q_{\alpha} (U)$ for these families are shown in 
Table~\ref{tab:transf} in terms of $\hat{g}( \texttt{Strength}, U)$. 


\begin{table}[thb]
\centering
\caption{Three different models for the generation of
discrete random variables,
for different scales of $\alpha$ of the recorded data \texttt{Strength}.
} 
\label{tab:transf}
\begin{tabular}{ l | l | c || c}
  \hline			
  \texttt{Strength} ($c>0$) & noise 
  & model & $\hat{g}( \texttt{Strength}, U) $ \\ \hline
  $s = c\log \alpha + d$ & $G \cong $ Gumbel  &  $s + c\cdot G$ 
  &
  $\texttt{Strength} - c \log (-\log U)$ \\
  $s = d\alpha^c $ & $G \cong $ Fr\'echet   &  $|s|  \cdot G^c$ &
  $|\texttt{Strength}| \cdot (-\log U)^{-c} $ \\
  $s = d\alpha^{-c} $ & $G \cong $ Neg.Exp.  &  $-|s|  \cdot G^{c}$   &
  $-|\texttt{Strength}| \cdot (-\log U)^{c} $ \\ \hline
\end{tabular}
\end{table}

\subsubsection{Gumbel family, Type 1} \label{sec:Gumb1}
The quantile function $Q(u) = -\log( -\log (u) )$ refers to the cumulative distribution 
$F(t) = e^{-e^{-t}}\mathbbm{1}_{(0,\infty)}(t)$ of the standard Gumbel distribution. 
In this case $X_\alpha = Q(\sqrt[\alpha]{U}) = \log \alpha + G $, where $G$ is a
standard Gumbel distribution, is a Gumbel distribution with mode $\log \alpha$. 

The Gumbel family $\{  e^{-\alpha e^{-t}} \mathbbm{1}_{(0,\infty)}(t), \alpha > 0\}$ 
is essentially the unique $\max$-compatible family with additive noise,
as the following theorem states. 
\begin{theorem}[Additive noise]\label{thm:add_noise}
The $\max$-compatible families with additive noise, i.e.\ 
where $X_\alpha = f(\alpha) + Q_1(U)$, are of the form
\[
X_\alpha = c ( \log(\alpha) + d ) + c G, \qquad c>0, d\in \mathbb{R}, G \cong \text{Gumbel}.
\]
\end{theorem}
This characterization may be immediately extended to the context of discrete choice models.
\begin{corollary}[Characterization of additive discrete choice model]
The probit model of the Example~\ref{exa:dcm} is the unique discrete choice model for which
${g} ( v,u) = v+u$, and, in this case $f(t,x_j) = c\log(p_{t\,j}) + d$, where $c>0$ and 
$d$ are real constant.
This means that, if the law of $\epsilon$ is not a Gumbel distribution, then there does not
exists a function $f = f(p)$ for which 
the utilities $U_{\alpha}= f(p_\alpha) + c\epsilon$ are generated with a $\max$-compatible family
and \eqref{eq:probit} holds.
\end{corollary}

\subsubsection{Gumbel family, Type 2}
If we substitute in \ref{rap:d} of Theorem~\ref{thm:main_repr} 
the cumulative distribution function of a Fr\'echet distribution
$F_1(t) = e^{-\tfrac{1}{t}} \mathbbm{1}_{(0,\infty)}(t)$,
the $\max$-compatible family that we obtain is the Type-2 Gumbel 
distribution family $\{F_\alpha(t) = e^{-\tfrac{\alpha}{t}} \mathbbm{1}_{(0,\infty)}(t), \alpha>0\}$.
The quantile function that generates the Fr\'echet distribution is of the form
 $Q_1(u) = -\tfrac{1}{\log (u) }$. The notable thing is that the generation of
 $X_\alpha\cong F_\alpha$ is done proportionally to $\alpha$:
 $X_\alpha = \alpha (  -\tfrac{1}{\log (U)} )$.

\subsubsection{Negative Exponential distribution}

When $X_\alpha = -\frac{\log(U)}{\alpha} $
is distributed as a negative Exponential distribution with parameter $\alpha$, then 
$F_\alpha(t) = (1-\exp(-\alpha t)) \mathbbm{1}_{(0,\infty)}(t)$.
Note that, by  \ref{rap:d} of Theorem~\ref{thm:main_repr}, 
\[
1-F_\alpha(-t)  = (\exp(t) \mathbbm{1}_{(-\infty ,0)}(t))^\alpha 
\] is a
$\max$-compatible family, and hence the 
Exponential distribution family 
is a 
$\min$-compatible family
by \ref{rap:b} of Theorem~\ref{thm:main_repr}.

The next theorem characterize the $\max$-compatible families with multiplicative noise, 
in terms of the last two $\max$-compatible families seen above. 
\begin{theorem}[Multiplicative noise]\label{thm:char_dep}
The $\max$-compatible families with multiplicative noise, i.e.\ 
where $X_\alpha = f(\alpha) Q_1(U)$, are of the form
\[
X_\alpha = d \alpha ^ c G^c, \qquad c, d \neq 0.
\]
In addition,
\begin{enumerate}
\item
if $c>0$, then $d>0$ and $G \cong \text{Fr\'echet}$;
\item
if $c<0$, then $d<0$ and $G \cong \text{Exponential}$.
\end{enumerate}
In particular, when $c=d=1$, 
$\mathcal{C}=\{  e^{-\tfrac{\alpha c }{t}} \mathbbm{1}_{(0,\infty)}(t), \alpha > 0\}$
is the Type-2 Gumbel family, 
and in this case $f(\alpha) = \alpha$.
When $c=d=-1$, 
$\mathcal{C}=\{  (\exp(t) \mathbbm{1}_{(-\infty ,0)}(t))^\alpha , \alpha > 0\}$
is the opposite of a exponential family, 
and in this case $f(\alpha) = \alpha^{-1}$.

\end{theorem}
As in Section~\ref{sec:Gumb1}, this result leads imediately to a characterization 
in the context of discrete choice models.
\begin{corollary}[Characterization of multiplicative discrete choice model]
The unique discrete choice models for which
${g} ( v,u) = v u$ are given by the consumer utility laws 
$U_{t\,j}^{(1)}= p_{t\,j} \epsilon_1$ or $U_{t\,j}^{(2)}= -1/(U_{t\,j}^{(1)})$,
where $\epsilon$ is a Type-2 Gumbel distributed random variable.
This means that, if the law of $\epsilon$ is not a Type-2 Gumbel distribution
or an exponential distribution, then there does not
exists a function $f = f(p)$ for which 
the utilities $U_{\alpha}= f(p_\alpha) \epsilon$ are generated with a $\max$-compatible family.
\end{corollary}

\subsection{Vademecum for model selection}

With the notation of Example~\ref{exa:dB},
when one records data with \texttt{Strength} 
that are proportional to the probability of their 
\texttt{QUAL} and bounded away from $0$,
the Negative Exponential distribution may be a good and simple choice. 
It should be preferred to the
Gumbel family, Type 2, for stability and precision in the simulation 
of the random variable, and hence ${\texttt{RND2}} =  - Y / \texttt{Strength}  $.
In \cite{Marsaglia00} it is discussed the ziggurat 
algorithm in simulating a Negative Exponential distribution, even if the direct method 
$Y = -\log(\texttt{RAND()})$ is usually preferable.

When one deals with self-information or surprisal, or
with a classifier that produces a score in logarithm space (as the multinomial
Bayes classifier in Example~\ref{exa:dB}), it is not convenient to exponentiate it, due to possible precision errors.
It is much more convenient to work with an additive model and Gumbel distributions of Type 1, see above.
We recall that the density of such a distribution is $f(t) = \exp(-(t+\exp(-t)))$, that means it has a
log-concave density, as the Negative Exponential distribution. 
Therefore, the random Gumbel variable can be generated either starting from a uniform distribution $U$ with
$Y = -\log(-\log(U))$ or with an appropriate direct method, as in \cite{Devroye12},
where a black-box style rejection method is proposed. 
Again, $\texttt{RND2} = \texttt{Strength} - \log(-\log(\texttt{RAND()}) )$ is a good choice in this case.

\section{Conclusions}\label{sec:concl}
In this paper, we propose a new class of parallelizable algorithms to simulate discrete random variables with general distributions.
The key idea is to increment simple operations that may be performed on each single possible outcome (local fully parallelizable operation), leaving to a single associative operation the final simulation of the random variable.

A probabilistic approach to this paradigm suggests future research. In fact, this algorithm selects the last index in the order statistics of the sample $(X_{\alpha_1},\ldots,X_{\alpha_n})$ generated with a $\max$-compatible family. But, while the theory of order statistics is highly developed (see, \cite{ORD03,adv06}), the theory of the ordered indexes of the order statistics merits to be exploited.

Besides this, the associative method that we have described in the previous sections suggests some research for a data structure which is optimal for the problem of the distributions that vary with time in scalable situations.

\appendix

We start by recalling and extending the notation given above.
$F,F_X,F_Y, F_{\alpha},\ldots$ denote cumulative 
distributions on $\mathbb R$,
while $X,X_{\alpha}, Y, Y_n, \ldots$ denote random variables on $\mathbb R$.
$X\sim Y$ means that $X$ and $Y$ share the same distribution, while 
$X\cong F$  means that the random variable $X$ has cumulative function $F$, also denoted
 by $F_X$. 
Thus, if $X\sim Y\cong F$, then $F(t_0) = P(X\leq t_0)=P(Y\leq t_0)$ 
and $F(t_0^{-}) = \lim_{s\uparrow t_0} F(s) = P(X< t_0)$ for any $t_0\in\mathbb{R}$. 
In addition, since $P(X = t_0) = F(t_0) - F(t_0^{-}) $ for any $t_0$, 
the continuity of $F$ at $t_0$ is equivalent to say that
$X$ does not have an atom at $t_0$. 
$U$ denotes always the random variable with uniform distribution on $(0,1)$:
$F_U(t) = \max(\min(t,1),0) $.
We will denote by $Q:(0,1) \to \mathbb{R}$ the quantile function associate to
a cumulative function $F$ in the following way:
\[
Q(u) = \inf\{ x \in\mathbb{R}\colon F(x)>u \} = \sup\{ y \in\mathbb{R}\colon F(y)\leq u \} .
\]
It is well known that, if $U$ is a $(0,1)$-uniform distributed random variable, then 
$Q(U) \cong F$. In addition, if $F$ is a continuous function, then $F(Q(u)) = u$ for any $u\in(0,1)$,
and, moreover, $Q(F(t)) = t$ $F$-almost everywhere.

\section{Basic results from probability theory}
The first lemma is a simple exercise of probability theory. We give here the proof for the sake of completeness.

\begin{lemma}\label{lem:XYatoms}
Let $X,Y$ be independent random variables
with common cumulative function $F$. Then $F$ is continuous if and only if $P(X=Y)=0$.
\end{lemma}
\begin{proof}
Assume that $X$ has an atom at $t_0$. Then 
\[
P(X=Y) \geq P(X=Y=t_0) = P(X=t_0) P(Y=t_0) = (P(X=t_0) )^2>0 .
\]
Conversely, if $F_{X} (\{y\})=0 $ for any $y$, by Fubini's Theorem,
\[
P(X=Y) = 
\int_{\mathbb R} \Bigg( \int_{\{y\}} dF_{X} (x) \Bigg) dF_{Y} (y) = 
\int_{\mathbb R} 0\, dF_{Y} (y)  =0. \qedhere
\]
\end{proof}

Given a cumulative function $F$, it is well known that the quantile function $Q:(0,1) \to \mathbb{R}$
induces the pushforward measure 
with cumulative distribution function $F$. 
When $F$ is continuous, $F(Q(u))=u$, and hence for any couple of measurable functions 
$G:\mathbb{R}\to [0,1],
h: [0,1]\times [0,1]\to [0,1]$, the change-of-variables formula (see \cite[Section 3.6]{Bog07}) gives
\begin{equation}\label{eq:ch_formula}
\int_{\mathbb{R}} h(G(t),F(t)) \, dF(t) 
=\int_{0}^1 h(G(Q(u)),u) \, du .
\end{equation}

We have the following result.
\begin{theorem}\label{thm:inversion}
Let $F,G$ be two continuous cumulative functions.
If, for any $n\geq 0$,
\[
\int_{\mathbb{R}} G(t) \, (F(t) )^n \, dF(t) = \frac{1}{n+2},
\]
then $G(t)= F(t)$.
\end{theorem}
\begin{proof}
Let $Q$ be the quantile function of $F$; we denote by $k:[0,1]\to[0,1]$ the measurable function
define by $k(u) = G(Q(u))$. 
Since $F$ is continuous, $Q(F(x)) = x$ for $F$-almost any $x$, and hence $G(x) = k(F(x))$. The thesis is
then proved once we show that $k(u) = u$ almost everywhere.

Now, it is well known that the coefficients $a^{(N)}_0,a^{(N)}_1, \ldots, a^{(N)}_N$ 
of the best $L^2$-polynomial approximation 
$P^{(N)}(u) = \sum_{j=0}^N a^{(N)}_j u^j$ on $(0,1)$ of the bounded measurable function $k(u)$ 
may be obtained by solving the following system:
\begin{equation}\label{eq:approx}
\sum_{j=0}^N a^{(N)}_j \int_{0}^1 u^{n+j} du = \int_{0}^1 k(u) u^n du, \qquad n=0, \ldots, N.
\end{equation}
As direct consequence of the approximation, $P^{(N)}(u) \to k(u)$ in $L^2(0,1)$. 

By \eqref{eq:ch_formula}, for any $n\geq 0$,
\[
\int_{0}^1 k(u) u^n du 
= \int_{0}^1 G(Q(u)) u^n du 
= \int_{\mathbb{R}} G(t) \, (F(t) )^n \, dF(t) = 
\frac{1}{n+2}.
\]
Since \(  \frac{1}{n+2} = 
\int_{0}^1 u^{n+1} du ,
\)
then the solution of \eqref{eq:approx} is 
\(a^{(N)}_j = \mathbbm{1}_{1}(j)\) or, equivalently, $ P^{(N)}(u) = u$. Then 
\[
k(u) = {\lim}_N P^{(N)}(u)  = {\lim}_N u = u. \qedhere
\]
\end{proof}

\section{Proof of the main results}
We now give a first property of any $\max$-compatible family that is
required in Remark~\ref{rem:cont}.
\begin{lemma}\label{lem:Falpha_cont}
Let  $\mathcal C= \{F_\alpha,\alpha > 0\}$ be a $\max$-compatible family. Then 
all the cumulative distribution functions $F_{\alpha}$ are continuous.
\end{lemma}

\begin{proof}
Let $X,Y$ be independent random variables
with common distribution function $F_{\alpha}$. Since, by \eqref{eq:def_max-compatible}, 
$P(X>Y)=1/2=P(Y>X)$, then $P(X=Y)=0$. The thesis follows by Lemma~\ref{lem:XYatoms}. 
\end{proof}
It is well known that if $X\cong F_X$ is a random variable with continuous distribution function,
then the random variable $Y=F_{X}(X)$ has a uniform distribution on $(0, 1)$.
The following corollary is an immediate consequence of Lemma~\ref{lem:Falpha_cont}. 
\begin{corollary}\label{cor:Falpha_to_U}
Let  $\mathcal C= \{F_\alpha,\alpha > 0\}$ be a $\max$-compatible family. 
If $X\cong F_{\alpha}$, then $F_{\alpha}(X) \sim U$. 
\end{corollary}

We prove the main result of the paper.
\begin{proof}[Proof of Theorem~\ref{thm:main_repr}]
Let $h$ be a monotone increasing function. Since $x<y \iff h(x) < h(y)$, then 
\begin{align*}
\{X_{\alpha_1} >\max(X_{\alpha_2},\ldots,X_{\alpha_n})\}
& =
\{h(X_{\alpha_1}) >h(\max(X_{\alpha_2},\ldots,X_{\alpha_n}))\}
\\& 
=
\{h(X_{\alpha_1}) >\max(h(X_{\alpha_2}),\ldots,h(X_{\alpha_n}))\}
\end{align*}
and hence $\ref{rap:a} \Longrightarrow \ref{rap:a1}$. The converse is trivial, since $h(t)=t$ is
a monotone increasing function.

Since $F_{-X}(t) = 1 - F_{X}(-t^{-})$, $F_X$ is continuous, and
$-\max(x,y)= \min(-x,-y)$, then $\ref{rap:a} \iff \ref{rap:b}$.

To prove that $\ref{rap:c} \Longrightarrow \ref{rap:d}$ we prove that, for ant fixed $t$,
$F_{\alpha}(t) = T^\alpha$, with $T = F_{1}(t)$. 
Accordingly, let $t$ be fixed, and define $f(\alpha) = F_\alpha (t)$.
By \ref{rap:c}, for any $\alpha_1,\alpha_2>0$, we have that
\begin{multline*}
f(\alpha_1)f(\alpha_2) = F_{\alpha_1}(t)
F_{\alpha_2}(t) = F_{X_1}(t)
F_{X_2}(t) 
= F_{\max(X_{\alpha_1} , X_{\alpha_2})} (t)  = 
F_{\alpha_1+\alpha_2}(t) = f(\alpha_1+\alpha_2) ,
\end{multline*}
where the fourth equality follows specifically from \eqref{eq:def_compat}. The thesis $\ref{rap:c} \Longrightarrow \ref{rap:d}$
is hence a consequence of the fact
that $f(\alpha) = T^\alpha$ (with $T\geq 0$) is the solution to the functional equation 
$f(\alpha_1+\alpha_2) = f(\alpha_1)f(\alpha_2) $ with $\alpha_1,\alpha_2>0, f(1) = T$. 

The opposite implication $\ref{rap:c} \Longrightarrow \ref{rap:d}$ also holds true. In fact, if
$X_{\alpha_1} \cong F_{\alpha_1}$ and 
$X_{\alpha_2} \cong F_{\alpha_2}$ 
are independent random variables, then $F_{\alpha_1} (t)F_{\alpha_2} (t) = F_{\max(X_1,X_2)} (t)$.
By \ref{rap:d}, we immediately obtain
$$
F_{\alpha_1+\alpha_2} (t) = (F_1(t))^{\alpha_1+\alpha_2} = (F_1(t))^{\alpha_1} (F_1(t))^{\alpha_2} =  
F_{\alpha_1} (t)F_{\alpha_2} (t) = F_{\max(X_1,X_2)} (t),
$$ 
and hence $\ref{rap:c} \Longrightarrow \ref{rap:d}$.

We prove $\ref{rap:d} \Longrightarrow \ref{rap:a}$ by first noticing that, 
if $X_{\alpha_1} \cong F_{\alpha_1}$ and 
$X_{\alpha_2} \cong F_{\alpha_2}$ 
are independent random variables, then 
\begin{align*}
P( X_{\alpha_1} \leq X_{\alpha_2}) = &
\int_{\mathbb{R}} P( X_{\alpha_1} \leq t | X_{\alpha_2} = t)\, dF_{\alpha_2}(t) 
\\
& = \int_{\mathbb{R}} F_{\alpha_1}(t)\, dF_{\alpha_2}(t)
= \int_{\mathbb{R}} (F_1(t))^{\alpha_1} \, d (F_1(t)^{\alpha_2})
\\
& = \alpha_2 \int_{\mathbb{R}} (F_1(t))^{\alpha_1+\alpha_2-1} \, d (F_1(t))
\intertext{and, by \eqref{eq:ch_formula}, the continuity of $F_1$ implies that}
& = \alpha_2 \int_{0}^1 u^{\alpha_1+\alpha_2-1} \, d u
= \frac{\alpha_2}{\alpha_1+\alpha_2} .
\end{align*}
As a consequence, 
\begin{equation}\label{eq:dim1}
  P( X_{\alpha_1} > X_{\alpha_2}) = 1 - P( X_{\alpha_1} \leq X_{\alpha_2}) = 
  1 -\frac{\alpha_2}{\alpha_1+\alpha_2}=
\frac{\alpha_1}{\alpha_1+\alpha_2} .
\end{equation}
Now, let 
$X_{\alpha_i} \cong F_{\alpha_i}, i=1,\ldots,n$ 
be independent random variables. 
Let $\{Y_m \cong F_{\sum_{j=2}^m\alpha_j}, m = 3, \ldots, n \}$ 
be a family of independent 
random variables and independent of $\sigma(X_{\alpha_1},X_{\alpha_2},\ldots,X_{\alpha_{n}})$.
Since $\ref{rap:d} \Longrightarrow \ref{rap:c}$, we use the associative structure of $\max$ 
and \eqref{eq:def_compat} to obtain
$Y_n \sim \max(X_{\alpha_2},\ldots,X_{\alpha_{n}}) $ in the following way:
\begin{align*}
\max(X_{\alpha_2},\ldots,X_{\alpha_{n}}) & = \max(\max(X_{\alpha_2}, X_{\alpha_3}),\ldots,X_{\alpha_{n}})
\\ &
\sim \max(Y_{3},X_{\alpha_{4}},\ldots,X_{\alpha_{n}}) 
\\ &
\qquad = \max(\max(Y_{3},X_{\alpha_{4}}),\ldots,X_{\alpha_{n}}) 
\\ &
\sim \max(Y_{{4}},\ldots,X_{\alpha_{n}}) 
\\ &
\qquad = \quad \cdots
\\ & \sim Y_n .
\end{align*}
Then, by  \eqref{eq:dim1}, since $Y_n \sim X_{\sum_{i=2}^n\alpha_i} $ is independent of $X_{\alpha_1}$,
\[
P\big( X_{\alpha_1} >\max(X_{\alpha_2},\ldots,X_{\alpha_n}) \big) = 
P\big( X_{\alpha_1} > Y_n \big) = 
\frac{\alpha_1}{\alpha_1+\sum_{i=2}^n\alpha_i} 
=\frac{\alpha_1}{\sum_{i=1}^n\alpha_i} ,
\]
which is the thesis: $\ref{rap:d} \Longrightarrow \ref{rap:a}$.

Now assume \ref{rap:a}, let $\alpha_1,\alpha_2$ be fixed and let $X_{\alpha_1} \cong F_{\alpha_1}$,
$X_{\alpha_2} \cong F_{\alpha_2}$,
$X_{\alpha_1+\alpha_2} \cong F_{\alpha_1+\alpha_2}$ 
be independent random variables. 
We 
denote by $G$ the cumulative function of 
\(
\max(X_{\alpha_1} , X_{\alpha_2} ) 
\).
For $n\geq 0$, let $Y_1,\ldots, Y_n$ be independent random variables distributed as
$F_{\alpha_1+\alpha_2}$ and independent of $\sigma( X_{\alpha_1+\alpha_2} ,X_{\alpha_1} , X_{\alpha_2} )$. 
By \eqref{eq:def_max-compatible}, we have
\begin{multline*}
\frac{1}{n+2} =
\frac{(\alpha_1+\alpha_2)}{(\alpha_1+\alpha_2) + \alpha_1+\alpha_2+ n(\alpha_1+\alpha_2)} 
=P\big(X_{\alpha_1+\alpha_2} > \max(X_{\alpha_1} , X_{\alpha_2},Y_1,\ldots, Y_n ) \big) .
\end{multline*}
By Remark~\ref{rem:cont}, we change $>$ with $\geq$, obtaining
\begin{align*}
\frac{1}{n+2} & =
P\big(X_{\alpha_1+\alpha_2} \geq \max(X_{\alpha_1} , X_{\alpha_2},Y_1,\ldots, Y_n ) \big) 
\\ & 
= P(\max(X_{\alpha_1} , X_{\alpha_2},Y_1,\ldots, Y_n ) \leq X_{\alpha_1+\alpha_2}  ) 
\\ & 
= 
\int_{\mathbb{R}} P\big( \max(X_{\alpha_1} , X_{\alpha_2},Y_1,\ldots, Y_n )  
\leq t \big| X_{\alpha_1+\alpha_2} = t\big)\, dF_{\alpha_1+\alpha_2}(t) 
\\ & 
= 
\int_{\mathbb{R}} P\big( \max( \max(X_{\alpha_1} , X_{\alpha_2}) , \max(Y_1,\ldots, Y_n ) )  \leq t 
\big)\, dF_{\alpha_1+\alpha_2}(t) 
\\ & 
= 
\int_{\mathbb{R}} G(t) \, (F_{\alpha_1+\alpha_2}(t) )^n \, dF_{\alpha_1+\alpha_2}(t)  .
\end{align*}
By Theorem~\ref{thm:inversion}, $G(t)= F_{\alpha_1+\alpha_2}(t)$. 
Since, by definition, $G(t)= F_{\max(X_{\alpha_1},X_{\alpha_2})}(t)$, then
$\ref{rap:a} \Longrightarrow \ref{rap:c}$.

To prove $\ref{rap:d} \iff \ref{rap:d2}$ it is sufficient to note that continuous distributions
are \emph{characterized} by strictly increasing quantile functions. Then, if we 
denote by $Q_1(u)$ the quantile function related to $F_1(t)$, we obtain
\[
P(Q_{\alpha}(U) \leq t) = P( Q_1(\sqrt[\alpha]{U}) \leq t) = 
P( \sqrt[\alpha]{U} \leq F_1(t) ) = 
(F_1(t))^{\alpha},
\]
that is the thesis.
\end{proof}

\begin{proof}[Proof of Theorem~\ref{thm:add_noise}]
Up to linear rescaling, we may assume that $f(1)=0$ and 
$Q_1(\tfrac{1}{e})=0$, that will simplify our computations
in the sequel. We recall that $F_\alpha$ is continuous, thus $F_\alpha(X_\alpha) \cong t\mathbbm{1}_{(0,1)}(t)$.
In addition, since $F_1$ and $Q_1$ are monotone functions, then
$F_1(f(\alpha) + Q_1(U)) ^\alpha = U$.
%
%
Now, since $F_1(f(\alpha) + Q_1(U)) = \exp(\frac{\log U}{\alpha})$, then 
by setting $U = \frac{1}{e}$, we obtain
\(
f(\alpha) = Q_1( \exp(-\tfrac{1}{\alpha}) )
\). Again, by substituting $ v = \exp(-\tfrac{1}{\alpha}) \in (0,1)$, 
\[
 F_1( Q_1(v) + Q_1(u)) = u ^{-\log v} = \exp( - \log u \log v), 
 \]
 which is equivalent to say that
\[ 
Q_1(v) + Q_1(u) = Q_1(\exp( - \log u \log v) ).
\]
Let $v = e^{-s}$, $u = e^{-t}$ ($s,t>0$), we get
\[ 
Q_1(e^{-s}) + Q_1(e^{-t}) = Q_1 (e^{ - st }),
\]
and hence, if $g(x) = Q_1(e^{-x})$, we obtain
\[ 
g(s) + g({t}) = g ({ st }), \qquad g(1) = 0,
\]
whose monotone continuous solutions are $g(x) = \pm \log (x)$. 
Since $Q_1$ is an increasing function, then $Q_1(u) = -\log (-\log u)$,
as expected.
\end{proof}

\begin{proof}[Proof of Theorem~\ref{thm:char_dep}]
When $X_\alpha \sim f(\alpha)G$, if we prove that all the $X_\alpha$'s must 
be either positive or negative with probability one, then the thesis will follow by 
applying Theorem~\ref{thm:add_noise} to $\log(X_\alpha)$ or 
$\log( - 1 / X_\alpha)$, respectively.

Since $G$ is continuous, let us denote by $p_+=P(G>0)$, $p_- = P(G<0)= 1-p_+$
and $p_0=\min(p_+,p_-)$. Let us assume by contradiction that $p_0>0$.

Now, we divide the indexes $\alpha$ according to the sign of $f$:
$\mathcal{F}_+ = \{\alpha>0\colon f(\alpha)>0\}$, 
$\mathcal{F}_- = \{\alpha>0\colon f(\alpha)<0\}$. 
We cannot have that $f(\alpha)=0$, since the generated distribution is not continuous,
contradicting  Lemma~\ref{lem:Falpha_cont}. Assume that both the sets 
$\mathcal{F}_+$ and $\mathcal{F}_-$ are not empty, then for each
$\alpha_+ \in \mathcal{F}_+$ and $\alpha_- \in \mathcal{F}_-$ we have
\begin{align*}
P( X_{\alpha_-} > X_{\alpha_+} ) & \geq 
P( X_{\alpha_-} > 0 ) P( X_{\alpha_+} < 0 ) \geq  
P( G < 0 )^2 \geq p_0^2;
\\
P( X_{\alpha_+} > X_{\alpha_-} ) & \geq 
P( X_{\alpha_+} > 0 ) P( X_{\alpha_-} < 0 ) \geq  
P( G > 0 )^2 \geq p_0^2;
\end{align*}
which is a contradiction with respect to \eqref{eq:def_max-compatible}, 
since at least one of the two sets $\mathcal{F}_+$ and $\mathcal{F}_-$ 
must be dense in a neighborhood of
$0$. Then, without loss of generality, we may assume that
$f(\alpha)>0$, for any $\alpha>0$. We have
\begin{align*}
P( X_{1} > X_{\alpha} ) & \geq 
P( X_1 > 0 ) P( X_{\alpha} < 0 ) \geq  
p_0^2;
\end{align*}
which is again a contradiction to \eqref{eq:def_max-compatible}
when $\alpha$ goes to $\infty$. Hence $p_0=0$.

Since $p_0=0$, we may assume that $p_+=1$. Again, 
$\alpha_+ \in \mathcal{F}_+$ and $\alpha_- \in \mathcal{F}_-$ imply
\( P(X_{\alpha_-} > X_{\alpha_+}) = 0\), that is contradictory to
\eqref{eq:def_max-compatible}. The thesis follows.
\end{proof}


\begin{thebibliography}{10}

\bibitem{adv06}
N.~Balakrishnan, E.~Castillo, and J.~M.~a. Sarabia, editors.
\newblock {\em Advances in distribution theory, order statistics, and
  inference}.
\newblock Statistics for Industry and Technology. Birkh\"auser Boston, Inc.,
  Boston, MA, 2006.
\newblock Selected papers from the International Conference on Distribution
  Theory, Order Statistics, and Inference held in honor of the 65th birthday of
  Barry C. Arnold at the University of Cantabria, Santander, June 16--18, 2004.

\bibitem{Bog07}
V.~I. Bogachev.
\newblock {\em Measure theory. {V}ol. {I}, {II}}.
\newblock Springer-Verlag, Berlin, 2007.

\bibitem{chen74}
H.~Chen.
\newblock On generating random variates from an empirical distribution.
\newblock {\em AIIE Transactions}, 6(2):163--166, 1974.
\newblock cited By 43.

\bibitem{ORD03}
H.~A. David and H.~N. Nagaraja.
\newblock {\em Order statistics}.
\newblock Wiley Series in Probability and Statistics. Wiley-Interscience [John
  Wiley \& Sons], Hoboken, NJ, third edition, 2003.

\bibitem{Hand01}
K.~Y. David J.~Hand.
\newblock Idiot's bayes: Not so stupid after all?
\newblock {\em International Statistical Review / Revue Internationale de
  Statistique}, 69(3):385--398, 2001.

\bibitem{Devroye86}
L.~Devroye.
\newblock {\em Non-Uniform Random Numbers Variate Generation}.
\newblock Springer-Verlag, New York, 1986.

\bibitem{Devroye12}
L.~Devroye.
\newblock A note on generating random variables with log-concave densities.
\newblock {\em Statistics and Probability Letters}, 82(5):1035 -- 1039, 2012.

\bibitem{KroPet79}
R.~A. Kronmal and A.~V. Peterson.
\newblock On the alias method for generating random variables from a discrete
  distribution.
\newblock {\em The American Statistician}, 33(4):214--218, 1979.

\bibitem{Fast2004}
G.~Marsaglia, W.~Tsang, and J.~Wang.
\newblock Fast generation of discrete random variables.
\newblock {\em Journal of Statistical Software}, 11:1--8, 2004.

\bibitem{Marsaglia00}
G.~Marsaglia and W.~W. Tsang.
\newblock The ziggurat method for generating random variables.
\newblock {\em Journal of Statistical Software}, 5(1):1--7, 2000.

\bibitem{Rennie03}
J.~D. Rennie, L.~Shih, J.~Teevan, D.~R. Karger, et~al.
\newblock Tackling the poor assumptions of naive bayes text classifiers.
\newblock In {\em ICML}, volume~3, pages 616--623. Washington DC, 2003.

\bibitem{Rub16}
R.~Rubinstein and D.~Kroese.
\newblock {\em Simulation and the Monte Carlo Method: Third Edition}.
\newblock 2016.
\newblock cited By 1.

\bibitem{DRNG_2013}
E.~Shmerling.
\newblock A range reduction method for generating discrete random variables.
\newblock {\em Statistics and Probability Letters}, 83(4):1094--1099, 2013.

\bibitem{Kenneth09}
K.~E. Train.
\newblock {\em Discrete Choice Methods with Simulation}.
\newblock Cambridge University Press, Cambridge, second edition, 2009.

\bibitem{walker74}
A.~Walker.
\newblock New fast method for generating discrete random numbers with arbitrary
  frequency distributions.
\newblock {\em Electronics Letters}, 10(8):127--128, 1974.

\bibitem{walker77}
A.~Walker.
\newblock An efficient method for generating discrete random variables with
  general distributions.
\newblock {\em ACM Transactions on Mathematical Software (TOMS)},
  3(3):253--256, 1977.
\newblock cited By 195.

\end{thebibliography}

\end{document}